\documentclass[11pt,a4paper]{article}

\usepackage[utf8]{inputenc}
\usepackage[english]{babel}

\usepackage{authblk}
\usepackage{geometry}
\usepackage{graphicx}
\usepackage{amssymb,amsmath,amsthm}
\usepackage{delarray}
\usepackage{enumitem}
\usepackage{tikz}
\usetikzlibrary{arrows.meta}
\usepackage{hyperref}
\graphicspath{{figures/}}

\newtheorem{theorem}{\textbf{Theorem}}

\newtheorem{corollary}{\textbf{Corollary}}
\newtheorem{remark}{\textbf{Remark}}

\newtheorem{remark*}{Remark}

\newcommand{\N}{\mathbb{N}}

\newcommand{\1}{\mathtt{1}}
\newcommand{\0}{\mathtt{0}}
\newcommand{\bool}{\{\0,\1\}}
\newcommand{\xor}{\oplus}
\newcommand{\ot}{\leftarrow}

\newcommand{\Poly}{{\mathsf{P}}}

\newcommand{\NP}{{\mathsf{NP}}}
\newcommand{\coNP}{{\mathsf{co\text{-}NP}}}
\newcommand{\NPNP}{{\mathsf{NP^{NP}}}}
\newcommand{\NPcoNP}{{\mathsf{NP^{co\text{-}NP}}}}
\newcommand{\BS}[1]{{\text{\textsc{BS}}_{#1}}}
\newcommand{\decisionpb}[3]{\fbox{\parbox{0.9\textwidth}{{\bf #1}\\{\it Input:} #2\\{\it Question:} #3}}}

\newcommand{\TODO}[1]{{\color{red}\fbox{TODO:} {\sf #1}}}

\title{Complexity of limit-cycle problems
in Boolean networks}
\author[1]{Florian Bridoux}
\author[1]{Caroline Gaze-Maillot}
\author[1,2]{K{\'e}vin Perrot}
\author[1]{Sylvain Sen{\'e}}
\affil[1]{Aix Marseille Univ., Univ. Toulon, CNRS, LIS, UMR 7020, Marseille, France.}
\affil[2]{Univ. C{\^o}te d'Azur, CNRS, I3S, UMR 7271, Sophia Antipolis, France.}
\date{}

\begin{document}
\setlist[itemize,enumerate]{nosep}
\maketitle

\begin{abstract}
  Boolean networks are a general model of interacting entities,
  with applications to biological phenomena such as gene regulation.
  Attractors play a central role, and the schedule of entities update
  is {\em a priori} unknown.
  This article presents results on the computational
  complexity of problems related to the existence of update schedules
  such that some limit-cycle lengths are possible or not.
  We first prove that given a Boolean network updated in parallel,
  knowing whether it
  has at least one limit-cycle of length $k$ is $\NP$-complete.
  Adding an
  existential quantification on the block-sequential update schedule
  does not change the complexity class of the problem,
  but the following
  alternation brings us one level above in the polynomial hierarchy: given a
  Boolean network, knowing whether there exists a block-sequential update
  schedule such that it has no limit-cycle of length $k$ is $\NPNP$-complete.
\end{abstract}

\section{Introduction}

Boolean networks (BNs) were introduced by McCulloch and Pitts in the 1940s 
through the well known formal neural networks~\cite{J-McCulloch1943} that 
are specific BNs governed by a multi-dimensional threshold function. 
Informally, BNs are finite dynamical systems in which entities having 
Boolean states may interact with each other over discrete time. After their 
introduction, neural networks were studied in depth from the mathematical 
standpoint. Among the main works on them are the introduction by Kleene of 
finite automata and regular expression~\cite{RR-Kleene1951}, first results  on 
the dynamical behaviors of linear feedback shift register~\cite{J-Huffman1959} 
and linear networks~\cite{J-Elspas1959}. These researches led Kauffman and 
Thomas (independently) from the end of the 1960s to develop the use of BNs in 
the context of biological networks 
modeling~\cite{J-Kauffman1969b,J-Thomas1973}, which has paved the way to 
numerous applied works at the interface between molecular biology, computer 
science and discrete mathematics. In parallel, theoretical developments were 
done in the framework of linear algebra and numerical analysis by 
Robert~\cite{J-Robert1969}, and in that of dynamical system theory and 
computational models, which constitutes the lens through which we look at BNs
in this paper.

In this context, numerous studies have already been led and have brought very 
important results. Considering that a BN can be defined as a collection of local 
Boolean functions (each of these defining the discrete evolution of one entity 
over time given the states of the entities that influence it), it can be 
represented by a directed graph at the static level, classically called the interaction 
digraph. Moreover, as a BN is by definition of finite size here, it is trivial to 
see that the trajectory of any of its configurations (or global state) ends into 
a cycle that can be a fixed point or a limit-cycle. The main theoretical 
objective in the domain is twofold: obtaining (combinatorial or algebraic) 
characterizations of the dynamics of such objects, through either their 
definition as collections of Boolean functions or their interaction graphs, 
and understanding the complexity of finding such characterizations. 

In these lines, Robert showed that retroaction cycles between entities in the 
interaction graph are necessary for a BN to have a non-trivial dynamical 
behavior~\cite{J-Robert1980} and Thomas conjectured strong relations between 
these retroaction cycles (well known as positive and negative cycles) and 
the existence of multi-stationarity (several fixed points) or 
limit-cycles~\cite{J-Thomas1981} which were proven 
later~\cite{J-Remy2008,J-Richard2007,J-Richard2010}. A notable fact about these 
seminal works is that they underline clearly that retroaction cycles are the 
engines of behavioral complexity (or dynamical richness). More recently, a real 
effort has been impulsed on the understanding of retroaction cycles. In 
particular, Demongeot et al. characterized exhaustively the behaviors of 
retroaction cycles and some of their intersections~\cite{J-Demongeot2012}. 
Furthermore, the problem of counting the number of fixed points and limit-cycles
has mushroomed. Advances have been done concerning fixed 
points~\cite{J-Aracena2004,J-Aracena2017}. Nevertheless, due to the high 
dependence of limit-cycle appearance according to the update schedule 
(\emph{i.e.} the way / order under which entities are updated over time), no 
general combinatorial results have been obtained, except for retroaction 
cycles~\cite{J-Demongeot2012}. In relation to complexity theory, the 
main known results based on BNs are: determining if a BN admits 
fixed points is NP-complete, counting fixed points is 
\#P-complete~\cite{J-Floreen1989,C-Orponen1992}, determining if a fixed point 
has a non-trivial attraction basin is NP-complete, determining if there exists 
another update schedule that conserves the limit-cycles of a given BN evolving 
in parallel is NP-hard~\cite{J-Aracena2013}. Moreover, a recent 
work~\cite{C-Bridoux2019} focused on related questions on fixed point complexity 
by focusing on interaction digraphs and not on BNs anymore (notice that several 
BNs admit the same interaction digraph).

In this paper, we impregnate from these last results and transfer the 
problematics to limit-cycles, which constitutes to our knowledge one of the 
first attempts to understanding limit-cycles from the complexity theory point of 
view with~\cite{J-Aracena2013,T-Gomez2015}. More precisely, considering that the 
input is a BN, we prove that determining if a BN evolving admits a limit-cycle 
of length $k$ is NP-complete whatever the update schedule (in the class of 
block-sequential updating modes, that is updating modes defined as ordered 
partitions of the set of entities). Furthermore, we show that determining if 
there exists a block-sequential update schedule such that a given BN admits no 
limit-cycles of length $k$ is NP$^{\text{NP}}$-complete.

In what follows, Section~\ref{sec:definitions} presents the main definitions 
that are used in the paper. Section~\ref{sec:sota} gives a brief state 
of the art of the problematic addressed. The main results of the paper are given 
in Section~\ref{sec:pb} and are followed by a conclusion developing some 
perspectives of this work.

\section{Definitions}
\label{sec:definitions}

We denote $\N_+$ the set of strictly positive integers, and $[n]=\{1,\dots,n\}$
for some $n \in \N_+$. For $x \in \bool^n$ and $i \in [n]$, we denote $x_i$ the
component $i$ of $x$, and $x+e_i$ the vector of $\bool^n$ obtained by flipping
component $i$ of $x$ (addition is performed modulo $2$). The symbol $\xor$ is
used for the binary operator {\em exclusive or} ({\em xor}).

\subsection{Boolean networks}

A {\em Boolean network} (BN) is a function $f:\bool^n \to \bool^n$,
that we see as $n$ {\em local functions} $f_1,\dots,f_n$ with $f_i:\bool^n \to
\bool$ for each $i \in [n]$. The {\em interaction digraph} of a BN $f$
captures the actual dependencies among its components, and is defined as
$G_f=(V,A)$, with $V=[n]$ and
$$(i,j) \in A \quad\iff\quad \exists x \in \bool^n : f_j(x) \neq f_j(x+e_i).$$
The arcs of the interaction digraph may be assigned signs $\sigma:A \to
\{+,-,\pm\}$ as follows:
\begin{itemize}
  \item $\sigma(i,j)=+$ when $\exists x \in \bool^n : x_i=0 \wedge f_j(x) > f_j(x+e_i)$,
  \item $\sigma(i,j)=-$ when $\exists x \in \bool^n : x_i=0 \wedge f_j(x) < f_j(x+e_i)$,
  \item $\sigma(i,j)=\pm$ when both conditions above hold.
\end{itemize}

For convenience, we may use various symbols to denote the components of the
network, but as it will always be a finite set a bijection with $[n]$ is
straightforward. The {\em size} of a BN is its number of components.

\subsection{Update schedules}

The {\em configuration space} is $\bool^n$, and it remains to explain how
components are updated. Given a BN $f$, a configuration $x$ and a subset $I
\subseteq [n]$, we denote\footnote{Parenthesis are used to differentiate update
schedules from iterations of a function.} $f^{(I)}(x)$ the configuration
obtained by updating components of $I$ only, {\em i.e.}
$$
  \text{for any } i \in [n],\quad f^{(I)}(x)_i
  = \left\{\begin{array}{ll}
    f_i(x) & \text{ if } i \in I\\
    x_i & \text{ otherwise.}
  \end{array}\right.
$$
Remark that $f^{([n])}=f$. A {\em block-sequential update schedule} is an
ordered partition of $[n]$, denoted $W=(W_1,\dots,W_t)$, and a BN $f$ updated
according to $W$ gives the deterministic discrete dynamical system on $\bool^n$
defined as
$$f^{(W)}=f^{(W_t)} \circ \dots \circ f^{(W_2)} \circ f^{(W_1)}.$$
The update schedule $([n])$ is called {\em parallel} (or {\em synchronous}).

\subsection{Attractors}

Given that the configuration space is finite and the dynamics is deterministic,
the orbit of any configuration convergences to a {\em fixed point} (a
configuration $x$ such that $f^{(W)}(x)=x$) or to a {\em limit-cycle} (a
configuration $x$ such that $(f^{(W)})^k(x)=x$ for some {\em length} $k \in
\N_+$, and such that $(f^{(W)})^{\ell}(x) \neq x$ for any $\ell \in [k-1]$). A
fixed point is a limit-cycle of length one, a limit-cycle is assimilated to any
of its configurations, and has a unique length.

Given a BN $f$, an update schedule $W$, and $k \in \N_+$, we denote
$\Phi_k(f^{(W)})$ the set of configurations in limit-cycles of length $k$,
{\em i.e.}
$$\Phi_k(f^{(W)})=\{ x \in \bool^n \mid (f^{(W)})^k(x)=x \text{ and } \forall\ 1 \leq \ell < k, (f^{(W)})^\ell(x) \neq x \}$$
and $\phi_k(f^{(W)})=\frac{|\Phi_k(f^{(W)})|}{k}$ the number of limit-cycles of
length $k$. Remark that for a fixed $k$, the quantity $\phi_k(f^{(W)})$ may
vary depending on 
$W$ (see Figure~\ref{fig:W}).

For retroaction cycles (such as those of Figure~\ref{fig:W}), the
dynamical behavior in terms of number of limit-cycles of size $k$,
whatever the update schedule, is entirely characterized
in~\cite{J-Demongeot2012} on the basis of~\cite{C-Goles2010}.

\begin{figure}
  \centerline{\includegraphics{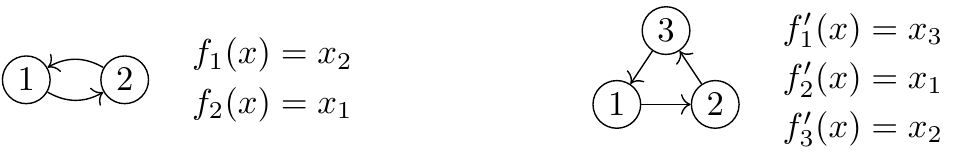}}
  \caption{Two BNs and their respective interaction digraphs (all arcs are
  positive). Left: for $W=(\{1\},\{2\})$ we have
  $\phi_2(f^{(W)})=0$, whereas for the parallel mode we have
  $\phi_2(f)=1$. Right: for $W'=(\{1\},\{2,3\})$ we have $\phi_2(f'^{(W')})=1$
  with $\0\0\1 \leftrightarrow \1\1\0$, whereas for the parallel mode we have
  $\phi_2(f')=0$.}
  \label{fig:W}
\end{figure}

\subsection{Problems}

\begin{remark}
  \label{remark:encoding}
  Note that an input BN $f$ is encoded with its local
  functions as propositional formulas
  (see also Remark~\ref{remark:truthtable} at the end).
\end{remark}

We are interested in the following decision problems related to attractors in
the dynamics of BNs, and especially limit-cycles.\\[.5em]
\decisionpb{$k$-limit-cycle problem ($k$-LC)}
{a BN $f$ updated in parallel.}
{does $\phi_k(f) \geq 1$?}

\begin{remark}
  \label{remark:k-limit-cycle-problem}
  $f$ updated in parallel is not a limitation here, since one can transform in
  polynomial time a BN $f$ and an updated schedule $W$ into a BN $f'$ updated
  in parallel such that $f'=f^{(W)}$ (simply construct local functions of $f'$
  from those of $f$ and $W$), as presented in~\cite{L-Robert1986}.
\end{remark}

\noindent
\decisionpb{Block-sequential $k$-limit-cycle problem (BS~$k$-LC)}
{a BN $f$.}
{does there exist $W$ block-sequential such that $\phi_k(f^{(W)}) \geq 1$?}
\\[.5em]
\decisionpb{Block-sequential no $k$-limit-cycle problem (BS~no~$k$-LC)}
{a BN $f$.}
{does there exist $W$ block-sequential such that $\phi_k(f^{(W)}) = 0$?}
\\[.5em]
Fixed points are invariant over block-sequential update schedules~\cite{L-Goles1990}, consequently 
{\bf $1$-LC} and {\bf BS~$1$-LC}
are identical.
However, the last two problems are not complement of each other,
because there exist some instance positive in both
(see Figure~\ref{fig:W} for an example).

For the reductions giving complexity lower bounds, we need the following
classical problems. For a formula $\psi$ on $\{\lambda_1,\dots,\lambda_n\}$ and
a partial assignment $v:\{\lambda_1,\dots,\lambda_s\} \to \bool$ for some $s
\in [n]$, we denote $\psi[v]$ the substitution $\psi[\lambda_1 \ot
v(\lambda_1),\dots,\lambda_s \ot v(\lambda_s)]$.
\\[.5em]
\decisionpb{3-SAT}
{a 3-CNF formula $\psi$ on $\{\lambda_1,\dots,\lambda_n\}$.}
{is $\psi$ satisfiable?}
\\[.5em]
\decisionpb{$\exists\forall$-3-SAT}
{a 3-CNF formula $\psi$ on $\{\lambda_1,\dots,\lambda_n\}$ and $s \in [n]$.}
{is there an assignment $v$ of $\lambda_1,\dots,\lambda_s$ such that\\
\phantom{\it Question:}
all assignments of $\lambda_{s+1},\dots,\lambda_n$ satisfy $\psi[v]$?}
\\[.5em]
{\bf 3-SAT} is a well known $\NP$-complete problem~\cite{J-Karp1972},
and {\bf $\exists\forall$-3-SAT} is $\NPNP$-complete~\cite{L-Papadimitriou1994}
(one level above in the polynomial hierarchy).
Also, note that $\NPNP=\NPcoNP$ since an oracle language or its complement are
equally useful.

\section{State of the art}
\label{sec:sota}

The {\bf $k$-limit-cycle problem} is known to be $\NP$-complete for 
$k=1$~\cite{J-Floreen1989}, and the fixed points of a BN are invariant for any 
block-sequential update schedule~\cite{L-Goles1990}. It has been proven 
in~\cite{J-Aracena2013b} that given a BN $f$, it is $\NP$-complete to know 
whether there exist two block-sequential update schedules $W,W'$ such
that $f^{(W)} \neq f^{(W')}$ (that is, they differ on at least one
configuration). This problem is indeed surprisingly difficult, but the proof
relies on a basic construction similar to Theorem~\ref{th:k-limit-cycle-problem} 
for $k=1$. More over, in~\cite{J-Aracena2013}, the authors study the 
computational complexity of limit cycle problems. Given a BN $f$, an update 
schedule $W$ and a limit-cycle $C$ of $f^{(W)}$, it is $\NP$-complete to know 
whether there exists another update schedule $W'$ (not equivalent to $W$) such 
that $f^{(W)}$ also has the limit-cycle $C$. Some variants of this problem are 
deduced to be $\NP$-complete: knowing whether the sets of limit cycles are 
equal, and whether the sets of limit-cycles share at least one element. This 
work focuses on finding block-sequential update schedules sharing limit cycles.
After writing this article, we learned that the PhD thesis of 
G\'omez~\cite{T-Gomez2015} contains results of a very close flavor: given a BN 
$f$, determining whether it is possible to find a block-sequential update 
schedule $W$ such that $f^{(W)}$ has at least one limit cycle (of any length 
greater than two) is $\NP$-complete, even when restricted to AND-OR networks. 
Moreover, the problem of finding a block-sequential $W$ such that $f^{(W)}$ has 
only fixed points is $\NP$-hard. In the sequel, we prove an analogous bound for 
the existence problem (Corollary~\ref{coro:BS-k-limit-cycle-problem}. Our 
construction also has only AND-OR local functions), and a stronger tight bound 
for the non-existence problem 
(Theorems~\ref{th:BS-no-k-limit-cycle-problem-hard-even}, 
\ref{th:BS-no-k-limit-cycle-problem-hard-k} and 
Corollary~\ref{coro:BS-no-k-limit-cycle-problem-hard-2}).
As a difference, in our setting the length of the limit-cycle is fixed
in the problem definition. It is also proven in~\cite{T-Gomez2015} that
given a BN $f$ and two configurations $x,y$, is there a $W$ such that $f^{(W)}
(x)=y$? is an $\NP$-complete problem.

Eventually, questions on the maximum number of fixed points possible when only 
the interaction digraph of a BN is provided, have already let some complexity
classes higher than $\NP$ appear in problems related to the attractors of 
BNs~\cite{C-Bridoux2019}.

\section{Complexity of limit-cycle problems}
\label{sec:pb}

The constructions presented in this section are gradually extended with more
involved arrangements of components, to prove complexity lower bounds from
formula satisfaction problems. The first result adapts a folklore proof for
fixed points (case $k=1$).

\begin{theorem}
  \label{th:k-limit-cycle-problem}
  {\bf $k$-LC} is $\NP$-complete for any $k \in \N_+$.
\end{theorem}

\begin{proof}
  The problem belongs to $\NP$ because one can check in polynomial time a
  certificate consisting of one configuration $x \in \bool^n$ of the
  limit-cycle of length $k$.
  Indeed, to check that $x \in \bool^n$ is in a limit-cycle of size $k$, it is
  sufficient to check that $f(x), \dots, f^{k-1}(x)$ are different from $x$ and
  that $f^k(x)$ equals $x$.

  To show that it is $\NP$-hard, we present a reduction from {\bf 3-SAT}. Given
  a 3-CNF formula $\psi$ on $\{\lambda_1,\dots,\lambda_n\}$ with $m$ clauses
  $C_1,\dots,C_m \in (\{\lambda_1,\dots,\lambda_n\} \cup
  \{\neg\lambda_1,\dots,\neg\lambda_n\})^3$, we construct the following BN of
  size $n+m+k$. The components are
  $$\{ \lambda_1,\dots,\lambda_n \} \cup \{ C_1,\dots,C_m \} \cup \{ \psi_1,\dots,\psi_k \}$$
  and the local functions are
  \begin{itemize}
    \item $f_{\lambda_i}(x)=x_{\lambda_i}$ for $i \in [n]$,
    \item $f_{C_j}(x)=\bigvee_{\lambda_i \in C_j} x_{\lambda_i} \vee \bigvee_{\neg\lambda_i \in C_j} \neg x_{\lambda_i}$ for $j \in [m]$,
    \item $f_{\psi_1}(x)=\neg x_{\psi_1} \wedge x_{\psi_k} \wedge ( x_{C_1} \wedge \dots \wedge x_{C_m} )$,
    \item $f_{\psi_i}(x)=\neg x_{\psi_i} \wedge x_{\psi_{i-1}}$ for $i \in \{2,\dots,k\}$.
  \end{itemize}
  If $k=1$, then we set $f_{\psi_1}(x)=\neg x_{\psi_1} \vee ( x_{C_1} \wedge
  \dots \wedge x_{C_m} )$. An example signed interaction digraph of this BN is
  presented on Figure~\ref{fig:k-limit-cycle}.
  
  The idea is that to get a limit-cycle of length $k$, one is forced to find in
  $x_{\lambda_1},\dots,x_{\lambda_n}$ an assignment satisfying $\psi$, in order
  to have $x_{C_j}=\1$ for all $j \in [m]$ and a configuration cycling through
  $x_{\phi_1},\dots,x_{\phi_k}$. Otherwise if
  $x_{\lambda_1},\dots,x_{\lambda_n}$ does not satisfy $\psi$, then the
  attractor is a fixed point (except for the case $k=1$). The articulation
  between the formula assignment and the limit-cycle of length $k$ hinges upon
  $f_{\psi_1}$.

  \begin{figure}
    \centerline{\includegraphics{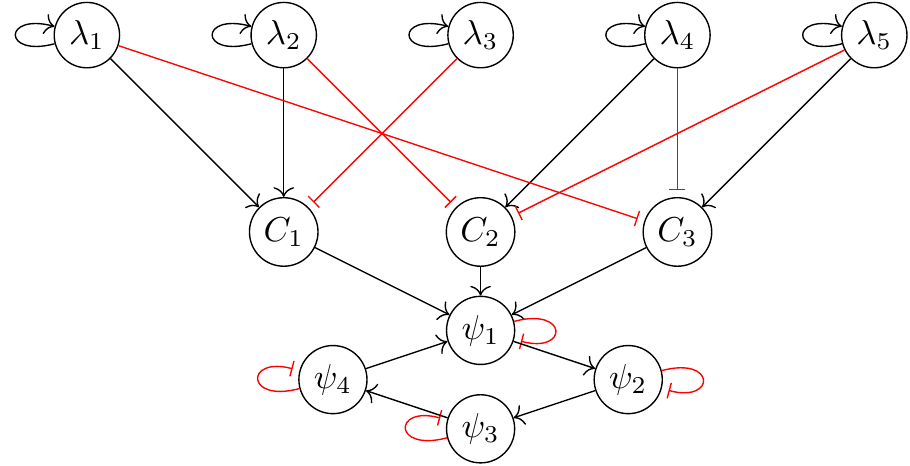}}
    \caption{Signed interaction digraph of the BN obtained for $k=4$ and the
    3-CNF formula $\psi=(\lambda_1 \vee \lambda_2 \vee \neg \lambda_3) \wedge
    (\neg \lambda_2 \vee \lambda_4 \vee \neg \lambda_5) \wedge (\neg \lambda_1
    \vee \neg \lambda_4 \vee \lambda_5)$. Negative arcs ($-$) are red with a
    flat head, positive arcs ($+$) are black (there are no $\pm$ arcs).}
    \label{fig:k-limit-cycle}
  \end{figure}

  \medskip

  Let us now prove that $\psi$ is satisfiable if and only if the BN has a
  limit-cycle of length $k$. Suppose $\psi$ is satisfied for
  $v:\{\lambda_1,\dots,\lambda_n\} \to \bool$, then the following configuration
  $x \in \bool^{n+m+k}$ is part of a limit-cycle of length $k$:
  \begin{itemize}
    \item $x_{\lambda_i}=v(\lambda_i)$ for all $i \in [n]$,
    \item $x_{C_j}=\1$ for all $j \in [m]$,
    \item $x_{\psi_1}=\1$ and $x_{\psi_2}= \dots =x_{\psi_k}=\0$.
  \end{itemize}
  Indeed, the state of components $\{\lambda_1,\dots,\lambda_n\} \cup
  \{C_1,\dots,C_m\}$ do not change, and the unique state $\1$ in the
  cycle\footnote{The cycle in the interaction digraph.} of components
  $\{\psi_1,\dots,\psi_k\}$ moves one component forward at each step (all other
  components being in state $\0$), and comes back to the initial configuration
  $x$ in $k$ steps, {\em i.e.} $f^k(x)=x$.

  \medskip

  For the reverse direction, suppose there is a limit-cycle of length $k$, and
  let $x$ be one of its configurations. Remark that in any attractor, the
  states of components $\{\lambda_1,\dots,\lambda_n\}$ are fixed, and so are
  the states of components $\{C_1,\dots,C_m\}$. As a consequence, in the local
  function $f_{\psi_1}$, the evaluation of the part $(x_{C_1} \wedge \dots
  \wedge x_{C_m})$ is fixed. For the sake of contradiction suppose that it is
  evaluated to $\0$, then so is $x_{\psi_1}$, then so is $x_{\psi_2}$, {\em
  etc}, and $x$ is a fixed point (in the case $k=1$ we have $f_{\psi_1}(x) \neq
  x_{\psi_1}$). Therefore, components $\{C_1,\dots,C_m\}$ are all in state $\1$,
  which, according to their local functions, is possible if and only if each of
  them has at least one of its predecessors in state $\1$ if it appears
  positively in the corresponding clause, or in state $\0$ if it appears
  negatively. As a conclusion the states of components
  $\{\lambda_1,\dots,\lambda_n\}$ in $x$ correspond to a valuation satisfying
  $\psi$.
\end{proof}

The second result initiates the consideration of update schedules in complexity
studies of the dynamics of BNs.
However, with an additional existential
quantification on the update schedule the problem remains $\NP$-complete (there
was already an existential quantification on configurations for the existence
of a limit-cycle), and it turns out that the same construction proves it.

\begin{corollary}
  \label{coro:BS-k-limit-cycle-problem}
  {\bf BS $k$-LC} is $\NP$-complete for any
  $k \in \N_+$.
\end{corollary}

\begin{proof}
  This problem still belongs to the class $\NP$, as one can check in polynomial time
  a certificate consisting of a block-sequential $W$ on $[n]$
  and one configuration $x \in \bool^n$ of the limit-cycle of length $k$.
  Indeed, it is sufficient to check that $f^{(W)}(x), \dots,
  {(f^{(W)})}^{k-1}(x)$ are different from $x$ and that ${(f^{(W)})}^{k}(x)$
  equals $x$.
 
  For the $\NP$-hardness we use the same construction as in the proof of
  Theorem~\ref{th:k-limit-cycle-problem}. Indeed, remark that the existential
  quantification on a block-sequential update schedule fits the reasoning. For
  the left to right direction of the {\em if and only if} we use the same $x$
  with $W=[n]$.
  And for the reverse direction, if $\psi$ is not satisfiable
  then for any block-sequential update schedule any configuration
  converges to a fixed point (the upper part
  is always fixed, and $x_{\psi_1}=\0$ fixes the cycle).
  %
  %
\end{proof}

We have seen in Theorem~\ref{th:k-limit-cycle-problem} and
Corollary~\ref{coro:BS-k-limit-cycle-problem} that with two consecutive
existential quantifications (one for a block-sequential update schedule and one
for a configuration of a limit-cycle) the problem remains in $\NP$. However,
{\bf BS no $k$-LC} corresponds to an existential
quantification (for a block-sequential update schedule) followed by a universal
quantification (for the absence of a limit-cycle). The next results therefore
jump one level above in the polynomial hierarchy.

\begin{theorem}
  \label{th:BS-no-k-limit-cycle-problem-in}
  {\bf BS no $k$-LC} is in $\NPNP$
  for any $k \in \N_+$.
\end{theorem}

\begin{proof}
  The problem belongs to the class $\NPcoNP=\NPNP$, as one can guess
  non-deterministically a block-sequential update schedule $W$ and then check
  in polynomial time (in $\NP$), using an oracle in $\coNP$, whether
  $\phi_k(f^{(W)})=0$. Once $W$ is fixed this last question is indeed in
  $\coNP$, as it is the complement of {\bf k-LC}, see
  Remark~\ref{remark:k-limit-cycle-problem} and
  Theorem~\ref{th:k-limit-cycle-problem}.
\end{proof}

The hardness proof is splitted into three results, developing
some incremental mechanisms and constructions.

\begin{theorem}
  \label{th:BS-no-k-limit-cycle-problem-hard-even}
  {\bf BS no $k$-LC} is
  $\NPNP$-hard
  for all $k$ even and strictly greater than $2$.
\end{theorem}

\begin{proof}
  We present a reduction from {\bf $\exists\forall$-3-SAT}.
  Given a 3-CNF formula $\psi$ on
  $\{\lambda_1,\dots,\lambda_n\}$ with $m$ clauses denoted as usual
  $C_1,\dots,C_m$, and an integer $s \in [n]$, we construct the following BN
  of size $2s+n+m+k+2$. The components are
  $$
    \{ \Omega,\psi \} \cup
    \{ \lambda_1,\dots,\lambda_n \} \cup
    \{ \lambda'_1,\dots,\lambda'_s \} \cup
    \{ \lambda''_1,\dots,\lambda''_s \} \cup
    \{ C_1,\dots,C_m \} \cup
    \{ \psi_0,\dots,\psi_{k-1} \}
  $$
  and the local functions are
  \begin{itemize}
    \item $f_\Omega(x)=\neg x_\Omega$,
    \item $f_{\lambda'_i}(x)=f_{\lambda''_i}(x)=x_\Omega$ for $i \in [s]$,
    \item $f_{\lambda_i}(x)=x_{\lambda'_i} \xor x_{\lambda''_i}$ for $i \in [s]$,
      and $f_{\lambda_i}(x)=x_{\lambda_i}$ for $i \in [n] \setminus [s]$,
    \item $f_{C_j}(x)=\bigvee_{\lambda_i \in C_j} x_{\lambda_i} \vee \bigvee_{\neg\lambda_i \in C_j} \neg x_{\lambda_i}$ for $j \in [m]$,
    \item $f_{\psi}(x)=x_{C_1} \wedge \dots \wedge x_{C_m}$,
    \item if $i \in \{0,\dots,k-1\}$ is even then $f_{\psi_i}(x)=
      \left\{\begin{array}{ll}
        x_{\psi_i} & \text{ if } x_{\psi}=\1 \vee x_\Omega=\0\\
        x_{\psi_{i-1 \mod k}} & \text{ otherwise}
      \end{array}\right.$,
    \item if $i \in \{0,\dots,k-1\}$ is odd then $f_{\psi_i}(x)=
      \left\{\begin{array}{ll}
        x_{\psi_i} & \text{ if } x_{\psi}=\1 \vee x_\Omega=\1\\
        x_{\psi_{i-1}} & \text{ otherwise}
      \end{array}\right.$.
  \end{itemize}
  An example signed interaction digraph of this BN is presented on
  Figure~\ref{fig:BS-no-k-limit-cycle-even}.

  The idea is that to prevent a possible limit-cycle of length $k$ to take
  place on components $\{\psi_0,\dots,\psi_{k-1}\}$, one is forced to solve
  the {\bf $\forall\exists$-3-SAT} instance and let $x_\psi=\1$ in any
  configuration $x$ that is part of an attractor. The existential variables are
  assigned in the block-sequential update schedule (on the updates of
  $\lambda_i$, $\lambda'_i$ and $\lambda''_i$ relative to the update of
  $\Omega$, for $i \in [s]$), and the universal variables all appear in both states in
  attractors (thanks to the positive loops on components
  $\{\lambda_{s+1},\dots,\lambda_{n}\}$).

  \begin{figure}
    \centerline{\includegraphics{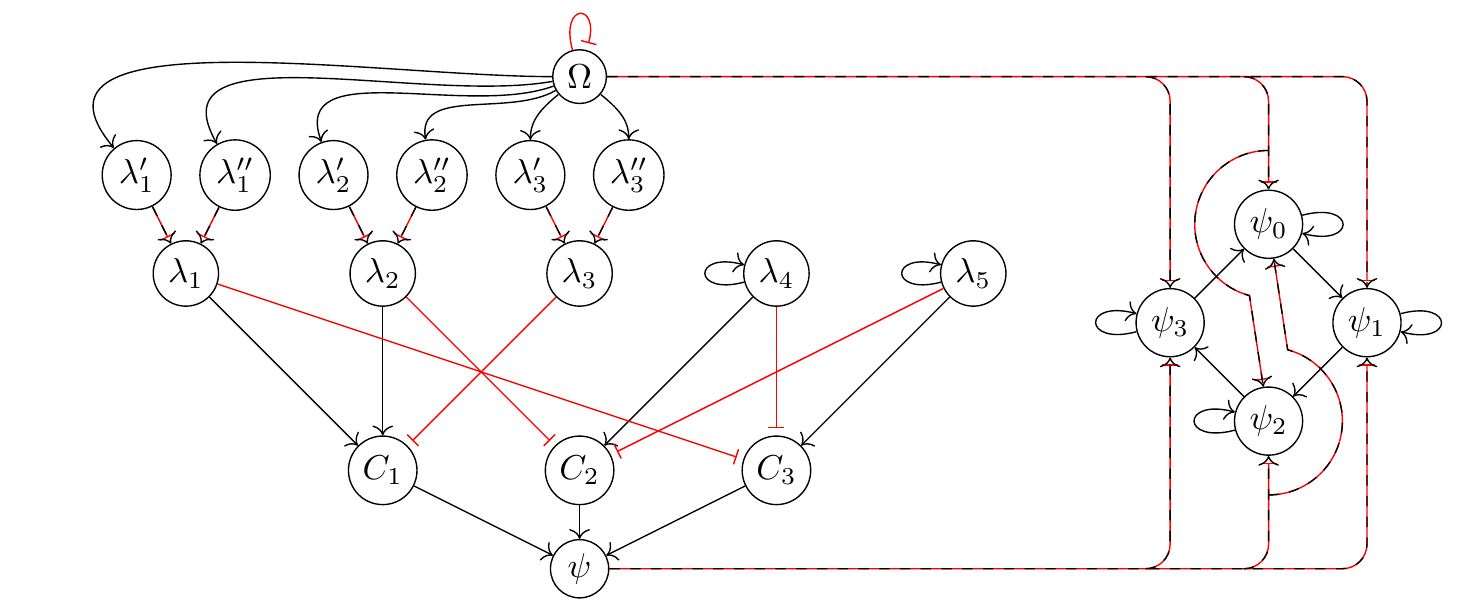}}
    \caption{Signed interaction digraph of the BN obtained for $k=4$, the
    3-CNF formula $\psi=(\lambda_1 \vee \lambda_2 \vee \neg \lambda_3) \wedge
    (\neg \lambda_2 \vee \lambda_4 \vee \neg \lambda_5) \wedge (\neg \lambda_1
    \vee \neg \lambda_4 \vee \lambda_5)$, and $s=3$. Negative arcs ($-$) are
    red with a flat head, positive arcs ($+$) are black, positive-negative arcs
    ($\pm$) are dashed with both colors and heads. Components $\Omega$ and
    $\psi$ are both connected to components $\psi_0$, $\psi_1$, $\psi_2$ and
    $\psi_3$ with arcs of sign $\pm$.}
    \label{fig:BS-no-k-limit-cycle-even}
  \end{figure}

  \medskip

  Let us now prove that there exists an assignment
  $v:\{\lambda_1,\dots,\lambda_s\} \to \bool$ such that all assignments
  $v':\{\lambda_{s+1},\dots,\lambda_n\} \to \bool$ verify $\psi[v][v'] \equiv
  \1$, if and only if there exists a block-sequential update schedule $W$ such
  that $f^{(W)}$ has a no limit-cycle of length $k$. Suppose there exists such
  an assignment $v$, then we define
  $$
    W=(T',\{\Omega\},F' \cup \{ \lambda''_1,\dots,\lambda''_s \},
    \{\lambda_1,\dots,\lambda_n\} \cup \{C_1,\dots,C_m\} \cup
    \{\psi\} \cup \{\psi_0,\dots,\psi_{k-1}\})
  $$
  with $T'=\{ \lambda'_i \mid v(\lambda_i)=\1 \}$ and $F'=\{ \lambda'_i \mid
  v(\lambda_i)=\0 \}$ (for $i \in [s]$). We claim that $f^{(W)}$ has no
  limit-cycle of length $k$. Indeed, the state of components
  $\{\lambda_1,\dots,\lambda_s\}$ correspond to the valuation $v$,
  because $\lambda'_i$ and $\lambda''_i$ for positive (resp. negative)
  variables are updated before and strictly after
  (resp. both strictly after) component $\Omega$ flips his
  state when it is updated, therefore are equal (resp. not equal) when
  local functions
  $f_{\lambda_i}$ compute their xor. The states of components $\{\Omega\}
  \cup \{\lambda'_1,\dots,\lambda'_s\} \cup \{\lambda''_1,\dots,\lambda''_s\}$
  all flip at each step (hence the required conditions on $k$), but the states
  of components $\{\lambda_1,\dots,\lambda_s\}$ are fixed. The states of
  components $\{\lambda_{s+1},\dots,\lambda_n\}$ are also fixed in any
  attractor, to arbitrary values among $\bool$. As $v$ satisfies $\psi$ for any
  valuation $v':\{\lambda_{s+1},\dots,\lambda_n\} \to \bool$, the states of
  components $\{C_1,\dots,C_m\}$ and $\psi$ are all fixed to $\1$ in any
  attractor. Hence, in any attractor we have:
  \begin{itemize}
    \item $\Omega$ flips its state at each time step,
    \item $\psi$ is fixed to state $\1$.
  \end{itemize}
  The local functions of components $\{\psi_0,\dots,\psi_{k-1}\}$ are designed
  to prevent any limit-cycle of length $k$ in this case: each of them is of the
  form $f_{\psi_i}(x)=x_{\psi_i}$, {\em i.e.} fixed. As a conclusion any
  attractor is in a limit-cycle of length $2 \neq k$.

  \medskip

  For the reverse direction we consider the contrapositive, suppose that there is
  no assignment $v:\{\lambda_1,\dots,\lambda_s\} \to \bool$ such that all
  assignments $v':\{\lambda_{s+1},\dots,\lambda_n\} \to \bool$ verify
  $\psi[v][v'] \equiv \1$. From what precedes, for any block-sequential update
  schedule $W$ there exists a configuration $x$ part of an attractor, with
  $x_{\lambda_{s+1}},\dots,x_{\lambda_{n}}$ chosen such that the state of
  $\psi$ is fixed to $\0$. Without loss of generality let use set
  $x_\Omega=\1$. Recall that in any attractor the states of components
  $\{\Omega\}\cup\{\lambda'_1,\dots,\lambda'_s\}\cup
  \{\lambda''_1,\dots,\lambda''_s\}$
  flip at each time step, the states of components
  $\{\lambda_1,\dots,\lambda_n\} \cup \{C_1,\dots,C_m\}$ are fixed, and that
  $k$ is even. Now if we let $x_{\psi_0}=x_{\psi_1}=\1$ and
  $x_{\psi_2}=\dots=x_{\psi_{k-1}}=\0$, then we claim that $x$ is in a
  limit-cycle of length $k$. We have to consider that in $W$, each $\psi_i$ may
  either be updated before $\Omega$, or strictly after $\Omega$, and we also
  have to consider the parity of $i$. According to local functions
  $f_{\psi_i}$, and because the state of component $\psi$ is fixed to $\1$, the
  four cases are as follows (recall that initially $x_\Omega=\1$):
  \begin{itemize}
    \item if $i$ is even and $\psi_i$ is updated before $\Omega$, then
      component $\psi_i$ copies the state of $\psi_{i-1 \mod k}$ at even time
      steps and is unchanged at odd time steps,
    \item if $i$ is even and $\psi_i$ is updated strictly after $\Omega$, then
      component $\psi_i$ copies the state of $\psi_{i-1 \mod k}$ at odd time
      steps and is unchanged at even time steps,
    \item if $i$ is odd and $\psi_i$ is updated before $\Omega$, then
      component $\psi_i$ copies the state of $\psi_{i-1}$ at odd time
      steps and is unchanged at even time steps,
    \item if $i$ is odd and $\psi_i$ is updated strictly after $\Omega$, then
      component $\psi_i$ copies the state of $\psi_{i-1}$ at even time
      steps and is unchanged at odd time steps.
  \end{itemize}
  Now observe that in any case, thanks to the parity of $i$ and the order of
  $\psi_i$ relative to component $\Omega$, when $\psi_i$ copies the state of
  $\psi_{i-1 \mod k}$, it is not possible that $\psi_{i-1 \mod k}$ has already
  copied the state of $\psi_{i-2 \mod k}$. As a consequence, at each time step
  the couple of states $\1$ moves one component forward along the cycle
  $\{\psi_0,\dots,\psi_{k-1}\}$, and after $k$ time steps we have
  $(f^{(W)}(x))^k=x$ (and not before).
\end{proof}

In the construction above, the fact that $f_\Omega(x)=\neg x_\Omega$ imposes
that any configuration converges to a limit-cycle of even length. Component 
$\Omega$ acts as a clock. For $k=2$ we can adapt the construction by letting
$x_\psi$ stop this clock when the formula is satisfied,
then in this case any configuration converges to a fixed point.

\begin{corollary}
  \label{coro:BS-no-k-limit-cycle-problem-hard-2}
  {\bf BS no $k$-LC} is
  $\NPNP$-hard
  for $k=2$.
\end{corollary}

\begin{proof}
  We present again a reduction from {\bf $\exists\forall$-3-SAT}, with a
  slightly modified construction from
  Theorem~\ref{th:BS-no-k-limit-cycle-problem-hard-even}.
  Given a 3-CNF formula $\psi$ on
  $\{\lambda_1,\dots,\lambda_n\}$ with $m$ clauses denoted as usual
  $C_1,\dots,C_m$, and an integer $s \in [n]$, we construct the following BN
  of size $2s+n+m+2$. The components are
  $$
    \{ \Omega,\psi \} \cup
    \{ \lambda_1,\dots,\lambda_n \} \cup
    \{ \lambda'_1,\dots,\lambda'_s \} \cup
    \{ \lambda''_1,\dots,\lambda''_s \} \cup
    \{ C_1,\dots,C_m \}
  $$
  and the local functions are
  \begin{itemize}
    \item $f_\Omega(x)=\neg x_\Omega \wedge \neg x_\psi$,
    \item $f_{\lambda'_i}(x)=f_{\lambda''_i}(x)=x_\Omega$ for $i \in [s]$,
    \item $f_{\lambda_i}(x)=x_{\lambda'_i} \xor x_{\lambda''_i}$ for $i \in [s]$,
      and $f_{\lambda_i}(x)=x_{\lambda_i}$ for $i \in [n] \setminus [s]$,
    \item $f_{C_j}(x)=\bigvee_{\lambda_i \in C_j} x_{\lambda_i} \vee \bigvee_{\neg\lambda_i \in C_j} \neg x_{\lambda_i}$ for $j \in [m]$,
    \item $f_{\psi}(x)= (x_{C_1} \wedge \dots \wedge x_{C_m}) \vee x_\psi$.
  \end{itemize}

  In this construction, the valuation of existential variables is still
  encoded in the block-sequential update schedule $W$, and all combinations
  of states on components corresponding to universal variables still appear
  in attractors.
  Now if the formula $\psi$ is a negative instance of
  {\bf $\exists\forall$-3-SAT}, then for any $W$ there exists a complete
  valuation (existential and universal variables) not satisfying the formula,
  hence in some attractor we have $x_\psi=\0$, and component $\Omega$
  flips at each step, giving a limit-cycle of length $2$.
  On the contrary, if $\psi$ is a positive instance of
  {\bf $\exists\forall$-3-SAT}, then there exists a $W$ such that
  all complete valuations satisfy the formula, hence in all attractors
  we have $x_\psi=\1$ (suppose $x_\psi=\0$, then it will converge
  to state $\1$ under update schedule $W$). Finaly, if $x_\psi=\1$
  then the attractor is a fixed point (it fixes component $\Omega$, then
  $\lambda''_i,\lambda'_i,\lambda_i$, then $C_j$), thus in this case there
  is no limit-cycle of length other than $1$.
\end{proof}

The idea presented in Corollary~\ref{coro:BS-no-k-limit-cycle-problem-hard-2} of
stopping a clock when the formula is satisfied (the clock gives a limit-cycle
of length $k$, and stopping it leads to a fixed point), can be extended to
any $k>2$. The challenge here is to design a clock giving a limit-cycle of length
$k$ for any block-sequential update schedule.

\begin{theorem}
  \label{th:BS-no-k-limit-cycle-problem-hard-k}
  {\bf BS no $k$-LC} is
  $\NPNP$-hard
  for any $k > 2$.
\end{theorem}

\begin{proof}
  The reduction is again from {\bf $\exists\forall$-3-SAT}.
  Given a 3-CNF formula $\psi$ on
  $\{\lambda_1,\dots,\lambda_n\}$ with $m$ clauses denoted as usual
  $C_1,\dots,C_m$, and an integer $s \in [n]$, we construct a BN
  of size $s+n+m+k+\lceil \log_2(\BS{k+1}) \rceil+3$
  with $\BS{k}$ the number of block-sequential update
  schedules\footnote{
    The number of block-sequential update schedules of size $k$
    equals the number of ordered partitions of a set of $k$ elements, also
    known as {\em ordered Bell number} (sequence \texttt{A000670} in the
    \texttt{OEIS} \cite{oeisA000670}). We have
    $$
      \BS{n}=
      \sum_{i=0}^k i! \left\{\begin{matrix} k \\ i \end{matrix}\right\}=
      \sum_{i=0}^k \sum_{j=0}^i (-1)^{i-j} \binom{i}{j}j^k
    $$
    using the Stirling numbers of the second kind (denoted with $\{\}$) counting the number of
    surjective maps from a set of $i$ elements to a set of $k$ 
    elements~\cite{R-Noual2011}.
  }
  of size $k$,
  on the components
  $$
    \begin{array}{c}
      \{ \Omega_0,\dots,\Omega_k \} \cup
      \{ \omega_1,\dots,\omega_{\lceil \log_2(\BS{k+1}) \rceil} \}
      \cup \{ \texttt{stop} \}\\[.2em]
      \cup~
      \{ \lambda_1,\dots,\lambda_n \} \cup
      \{ \lambda'_1,\dots,\lambda'_s \} \cup
      \{ C_1,\dots,C_m \} \cup
      \{ \psi \}.
    \end{array}
  $$
  Recall that $k$ is fixed in the problem definition, hence we do not need\footnote{
    \TODO{Can we nevertheless consider it, just for fun, in this footnote?}
  } to
  consider the growth of $\log_2(\BS{k+1})$, which is a constant from the point
  of view of {\bf BS no $k$-LC}.

  \medskip

  The idea is to build a clock of length $k$ on the $k+1$ components $\Omega$,
  with some $\1$ state moving forward at each step. However, it will not move
  forward from components $\Omega_i$
  to $\Omega_{i+1}$, {\em etc}
  modulo $k$, but instead it will move forward according to the order of
  components $\Omega$ in the current update schedule, which is supposed to
  be encoded (in binary\footnote{
    Since $k$ is a constant we can consider any computable encoding of the
    block-sequential update schedules,
    for example their numbering according to the lexicographic order
    (each subset of $\{\Omega_0,\dots,\Omega_k\}$ corresponds to a digit
    on $k$ bits).
  }) on components $\omega$ (positive loops on components $\omega$ will
  let them take any fixed value in attractors).
  Similarly to the construction
  of Theorem~\ref{th:BS-no-k-limit-cycle-problem-hard-even}, the update order
  of $\lambda_i,\lambda'_i$
  compared to clock component $\Omega_0$ encodes existential
  variables in $W$, and positive loops on universal variables let them take any
  fixed value in attractors. Finally, $x_\psi=\1$ will stop the clock. Regarding
  the logics of the proof, if $\psi$ is a positive instance then one can choose
  $W$ with components $\Omega$ updated in parallel
  and $\lambda'_i$ encoding the existential variables to satisfy $\psi$,
  then component $\psi$ will be in state $\1$ in any
  attractor (thanks to the construction, regardless of the update
  schedule encoded on components $\omega$)
  hence leading to fixed points only. If
  $\psi$ is a negative instance, then for any $W$ we can set components $\omega$
  accordingly to have a working clock of length $k$,
  and no matter the encoding of existential
  variables there exists a choice of states on components corresponding to
  universal variables such that $\psi$ is in state $\0$, letting the clock
  tick forever and create a limit-cycle of length $k$.

  \medskip

  The local functions are
  \begin{itemize}
    \item $f_\texttt{stop}(x)= x_\texttt{stop} \vee x_\psi \vee
      \texttt{error}(x_{\omega_1},\dots,x_{\omega_{\lceil \log_2(\BS{k+1}) \rceil}})$,
      where $\texttt{error}(\omega)$ equals $\1$ when components $\omega$
      do not encode a block-sequential update schedule,
    \item $f_{\omega_i}(x)= x_{\omega_i}$ for $i \in [\lceil \log_2(\BS{k+1}) \rceil]$,
    \item for the definition of $\Omega_i$, let us consider the update schedule
      encoded on components $\omega$ in some configuration $x$,
      and denote $j_0(x),\dots,j_k(x)$ the lexicographically
      minimal permutation of $0,\dots,k$ such that
      $\Omega_{j_0(x)} \preccurlyeq_{x_\omega} \Omega_{j_1(x)} \preccurlyeq_{x_\omega} \dots \preccurlyeq_{x_\omega} \Omega_{j_k(x)}$,
      where $a \preccurlyeq_{x_\omega} b$ means that component $a$ is updated prior to or
      simultaneously with component $b$ in the update schedule encoded
      on components $\omega$ in configuration $x$;\\
      for $i \in \{0,\dots,k\}$,
      $f_{\Omega_i}(x)= \neg x_\texttt{stop} \wedge
      \left\{\begin{array}{ll}
        \1 & \text{ if } i=j_p(x) \text{ and } x_{\Omega_{j_{p-1}(x)}}=\1 \text{ and }\\
          &\qquad \Big( x_\omega \neq (\{\Omega_0,\dots,\Omega_k\}) \text{ or } i \neq k \Big)\\
        \0 & \text{ otherwise,}
      \end{array}\right.$
      with $x_\omega$ the block-sequential update schedule encoded on components $\omega$,
    \item $f_{\lambda'_i}(x)=x_{\Omega_0}$ for $i \in [s]$,
    \item for $i \in [s]$,
      $f_{\lambda_i}(x)=
      \left\{\begin{array}{ll}
        x_{\Omega_0} \xor x_{\lambda'_i} & \text{ if } x_{\Omega_0}=\1\\
        x_{\lambda_i} & \text{ otherwise,}
      \end{array}\right.$
      \\
      and for $i \in [n] \setminus [s]$, $f_{\lambda_i}(x)=x_{\lambda_i}$,
    \item $f_{C_j}(x)=\bigvee_{\lambda_i \in C_j} x_{\lambda_i} \vee \bigvee_{\neg\lambda_i \in C_j} \neg x_{\lambda_i}$ for $j \in [m]$,
    \item $f_{\psi}(x)= (x_{C_1} \wedge \dots \wedge x_{C_m})$.
  \end{itemize}
  An example signed interaction digraph of this BN is presented on
  Figure~\ref{fig:BS-no-k-limit-cycle-k}.
  \begin{figure}
    \centerline{\includegraphics{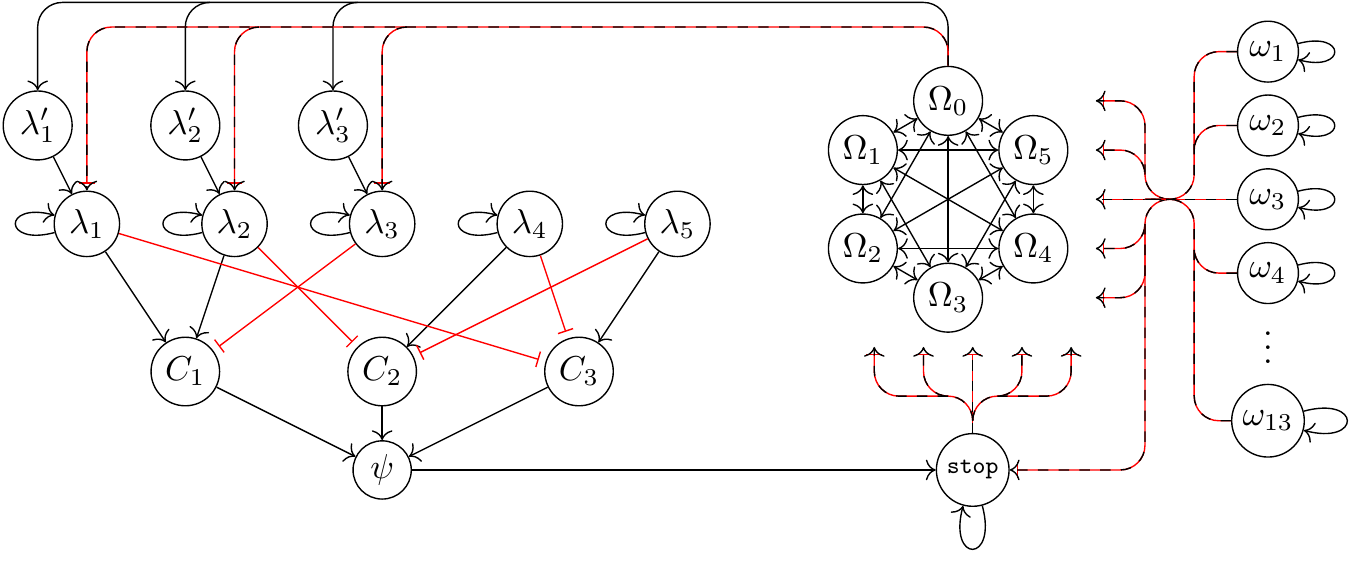}}
    \caption{Signed interaction digraph of the BN obtained for $k=5$ ($\BS{6}=4683$), the
    3-CNF formula $\psi=(\lambda_1 \vee \lambda_2 \vee \neg \lambda_3) \wedge
    (\neg \lambda_2 \vee \lambda_4 \vee \neg \lambda_5) \wedge (\neg \lambda_1
    \vee \neg \lambda_4 \vee \lambda_5)$, and $s=3$. Negative arcs ($-$) are
    red with a flat head, positive arcs ($+$) are black, positive-negative arcs
    ($\pm$) are dashed with both colors and heads. All components $\texttt{stop}$ and
    $\omega_1,\dots,\omega_{13}$ are connected to all components $\Omega_0,\dots,\Omega_5$
    with arcs of sign $\pm$.}
    \label{fig:BS-no-k-limit-cycle-k}
  \end{figure}
  First remark that if $x_\texttt{stop}=\1$ then $x$ converges to a fixed point
  (the clock stops), hence we will consider thereafter only attractors from
  configurations with component $\texttt{stop}$ in state $\0$.

  \medskip

  Suppose $\psi$ is a negative instance of {\bf $\exists\forall$-3-SAT}. For
  any block-sequential update schedule $W$, consider a configuration $x$ such
  that components $\omega$ encode the projection of $W$ on the clock components.
  If $W$ is not the parallel update
  schedule, the clock has the following dynamics (time goes downward,
  one step per line):
  $$
    \begin{array}{|c|c|c|c|c|c|c|}
      \hline
      \Omega_{j_k(x)} & \Omega_{j_{k-1}(x)} & \Omega_{j_{k-2}(x)} & \dots & \Omega_{j_2(x)} & \Omega_{j_1(x)} & \Omega_{j_0(x)}\\
      \hline
      \0 & \1 & \0 & \dots & \0 & \0 & \0\\
      \hline
      \0 & \0 & \1 & \dots & \0 & \0 & \0\\
      \hline
      \vdots & \vdots & \vdots & & \vdots & \vdots & \vdots\\
      \hline
      \0 & \0 & \0 & \dots & \1 & \0 & \0\\
      \hline
      \0 & \0 & \0 & \dots & \0 & \1 & \0\\
      \hline
      \1 & \0 & \0 & \dots & \0 & \0 & \1\\
      \hline
      \0 & \1 & \0 & \dots & \0 & \0 & \0\\
      \hline
    \end{array}
  $$
  and if $W$ is the parallel update schedule $(\{\Omega_0,\dots,\Omega_k\})$,
  the clock has the following dynamics (time goes downward,
  one step per line):
  $$
    \begin{array}{|c|c|c|c|c|c|c|}
      \hline
      \Omega_{0} & \Omega_{1} & \Omega_{2} & \dots & \Omega_{k-2} & \Omega_{k-1} & \Omega_{k}\\
      \hline
      \0 & \1 & \0 & \dots & \0 & \0 & \0\\
      \hline
      \0 & \0 & \1 & \dots & \0 & \0 & \0\\
      \hline
      \vdots & \vdots & \vdots & & \vdots & \vdots & \vdots\\
      \hline
      \0 & \0 & \0 & \dots & \1 & \0 & \0\\
      \hline
      \0 & \0 & \0 & \dots & \0 & \1 & \0\\
      \hline
      \1 & \0 & \0 & \dots & \0 & \0 & \0\\
      \hline
      \0 & \1 & \0 & \dots & \0 & \0 & \0\\
      \hline
    \end{array}
  $$
  thus we have a clock of length $k$ in any case:
  \begin{itemize}
    \item when $W$ is not parallel the minimum component according
      to $\preccurlyeq_{x_\omega}$ and the lexicographical order is
      skipped (the $\1$ state moves two components forward),
    \item when $W$ is parallel component $\Omega_k$ is discarded
      (it remains in state $\0$ and the clock ticks on components
      $\Omega_0,\dots,\Omega_{k-1}$).
  \end{itemize}
  Furthermore, for $i \in [s]$ component $\lambda'_i$
  goes to state $\1$ exactly once every $k$ steps, and the relative
  positions of components $\lambda'_i,\lambda_i,\Omega_0$ fixes the value of
  component $\lambda_i$:
  \begin{itemize}
    \item if ($\lambda'_i =_W \Omega_0$)
      or ($\Omega_0 \prec_W \lambda_i \preccurlyeq_W \lambda'_i$)
      or ($\lambda_i \preccurlyeq_W \lambda'_i \prec_W \Omega_0$)
      or ($\lambda'_i \prec_W \Omega_0 \prec_W \lambda_i$)
      then $x_{\lambda_i}=\1$,
    \item otherwise $x_{\lambda_i}=\0$.
  \end{itemize}
  Since the instance $\psi$ is negative, for any assignment of states to
  components $\lambda_1,\dots,\lambda_s$ (corresponding to existential
  variables), we can set the states of components
  $\lambda_{s+1},\dots,\lambda_n$ (corresponding to universal variables) so
  that at least one clause $C_j$ is not satisfied hence $x_{C_j}=\0$ and
  $x_\psi=0$. As a consequence $f_\texttt{stop}(x)=\0$, {\em i.e.} the clock is not
  stopped, and therefore it creates a limit-cycle of length $k$.

  \medskip

  Suppose $\psi$ is a positive instance of {\bf $\exists\forall$-3-SAT}, with
  $v:\{\lambda_1,\dots,\lambda_s\} \to \bool$ an assignment such that for all
  $v':\{\lambda_{s+1},\dots,\lambda_n\} \to \bool$ we have $\psi[v][v'] \equiv \1$.
  We define
  $$
    W=(T',\{\Omega_0,\dots,\Omega_k\},\mathcal{R}),
  $$
  with $T'=\{\lambda'_i \mid v(\lambda_i)=\1\}$ and $\mathcal{R}$ all the other
  components. We consider a case disjunction on the starting configuration.
  \begin{itemize}
    \item If components $\omega$ encode the parallel update schedule, then from
      what preceeds the states of components $\lambda_i$ for $i \in [s]$ encode
      $v$ and component $\psi$ will eventually be in state $\1$, so
      does component $\texttt{stop}$ and the clock stops, leading to a fixed
      point.
    \item If components $\omega$ do not encode the parallel update schedule, then
      from the definition of local function $f_{\Omega_i}$ we will have a clock
      of length $k+1$ (with state $\1$ moving one component forward at each
      step, in an order given by the update schedule encoded on components $\omega$).
      However, it
      does not alter the fact that the states of components $\lambda_i$ for $i
      \in [s]$ encode $v$, therefore the same deductions apply: the
      configuration converges to a fixed point.
  \end{itemize}
  We can conclude that under update schedule $W$, any configuration converges
  to a fixed point hence there is no limit-cycle of length $k$.
\end{proof}

\begin{remark}
  \label{remark:truthtable}
  Encoding local function as truth tables of the
  components it effectively depends on (its in-neighbors in the interaction
  digraph) would also lead to the same complexity results,
  because all the 
  constructions presented for hardness results can be adapted so that
  each component depends on
  a bounded number of components (the resulting interaction digraph has
  a bounded in-degree), given that $k$ is a constant.
\end{remark}

\section{Conclusion}

We have characterized precisely the computational complexity of problems
related to, given a BN, the existence or not of limit-cycles of some fixed
length $k$, with the quantifier alternation of ``does there exist an update
schedule such that all configurations are not in a limit-cycle of size $k$''
bringing us to level $\Sigma^\Poly_2$ of the polynomial hierarchy.

Remark that all the constructions presented in our reductions (except for
Theorem~\ref{theorem:3} which is subsumed by Theorem~\ref{theorem:4}) are such
that the resulting BN has either some limit cycles of size $k$, or only fixed
points.
Consequently, the same results directly hold for the problem $\phi_k(f^{(W)})$
is replaced by $\phi_{\geq k}(f^{(W)})=\sum_{\ell \geq k} \phi_\ell(f^{(W)})$,
{\em i.e.} we consider limit-cycles of length at least $k$ instead of exactly
$k$. With little additional work the proofs may also be adapted to $\phi_{\leq
k}(f^{(W)})=\sum_{\ell \leq k} \phi_\ell(f^{(W)})$, {\em i.e.} if we consider
limit-cycles of length at most $k$ (fixed points should be transformed into
limit-cycles of length larger that $k$).

Finally, if $k$ is part of the input, is there a drastic complexity increase as
observed for problems related to the number of fixed points
in~\cite{C-Bridoux2019}? The construction presented in the proof of
Theorem~\ref{th:BS-no-k-limit-cycle-problem-hard-k} makes heavy use of being a
$k$ constant.

We hope that these first results on the complexity of deciding the existence of
limit-cycles in Boolean networks opens a promising research direction,
confronting the necessary difficulty of considering a diversity of update
schedules. The lens of computational complexity reveals, via the gadgets
employed in lower bound constructions, mechanisms at the heart of Boolean
network's dynamical richness.

\section{Acknowledgments}

The authors are thankful to
project ANR-18-CE40-0002-01 ``FANs'',
project ECOS-CONICYT C16E01,
project STIC AmSud CoDANet 19-STIC-03 (Campus France 43478PD),
for their funding.

\bibliographystyle{plain}
\bibliography{biblio}

\end{document}